\documentclass{article}

\usepackage{arxiv}

\usepackage[utf8]{inputenc} 
\usepackage[T1]{fontenc}    
\usepackage{hyperref}       
\usepackage{url}            
\usepackage{booktabs}       
\usepackage{amsfonts}       
\usepackage{nicefrac}       
\usepackage{microtype}      
\usepackage{lipsum}

\usepackage{graphicx}
\usepackage{mathtools}
\usepackage{amsmath}
\usepackage{amsthm}
\usepackage{amssymb}
\usepackage{subfig}

\usepackage{algorithm2e}

\DeclareMathOperator*{\mini}{min}

\newtheorem{prop}{Proposition}

\title{Enhancing Quality of Experience using DHT Overlay on Device-to-Device Communications in LTE-A Networks}

\author{
  Sumit Kumar Tetarave \\
  Department of Computer Science \& Engineering\\
  Indian Institute of Technology Patna\\
  Bihar, India 801103 \\
  \texttt{sktetarave@iitp.ac.in} \\
   \And
 Somanath Tripathy \\
  Department of Computer Science \& Engineering\\
  Indian Institute of Technology Patna\\
  Bihar, India 801103 \\
  \texttt{som@iitp.ac.in} \\
     \And
 R. K. ~Ghosh \\
  Department of Computer Science \& Engineering\\
  Indian Institute of Technology Kanpur\\
  UP, India 208016 \\
  \texttt{rkg@cse.iitk.ac.in} \\
}

\begin{document}
\maketitle

\begin{abstract}
In~\cite{tetarave2018v}, we proposed a two-tier model for creating and maintaining a distributed hash table (DHT) overlay of smartphones. The purpose was to reduce Internet usage and increase D2D content sharing over LTE-A/5G network. However, the real challenge is not in creating a two-tier model, but to evolve an efficient overlay that offers enhanced opportunities for device-to-device (D2D) content sharing over the existing model. In this paper, we formulated the selection of tier-1 devices as finding P-medians in a distribution network and solved it using Lagrangian relaxation method. Tier-2 devices become clients seeking content sharing services from tier-1 devices.  A strong motivation in this work is to raise a user's perception of the grade of service known as Quality of Experience (QoE). We analyzed the QoE assessment challenge for resource-constrained smartphones under the proposed model of enhanced D2D communication. Our focus is to establish a framework to assess QoE for apps and services over LTE-A/5G networks with an improved level of D2D communication.  The simulation and the experimental results validate the claim that substantial improvements in QoE are indeed possible with the proposed mathematical model for selection and placement of tier-1 mobile devices and maintaining a DHT for D2D communication. 
\end{abstract}


\section{Introduction}
\label{intro}
The mobile networks have witnessed exponential growth during the last decade or so. The digitization process is extremely convenient for any transaction involving finance, business, health-care system, education, travels, entertainments, etc. However, most of these transactions are also associated with the online sharing of digital content on smartphones either over the Internet or the cellular networks (LTE/5G/3G, etc.). Excessive use of a phone's wireless interface leads to substantial cost and the drainage of the battery power~\cite{nika2015energy, carroll2010analysis}. Moreover, the sharing of data becomes difficult due to poor cellular network signals. 

WiFi, device-to-device (D2D), or Bluetooth are inexpensive wireless communication channels which can be used under a limited coverage area (vicinity) for content sharing applications~\cite{mcnamara2008creating, paiva2015intuitive, liu2011efficient}.
The performance of the content sharing applications can be improved further by creating and maintaining a distributed hash table (DHT) based overlay of mobile devices in a cellular network such as LTE-A~\cite{tetarave2018v, cchord2012}. In a DHT based implementation, the members of different vicinity communicate using only a few cellular radio links (or the Internet). Therefore, the users can use different apps and services with a better QoE than the non-DHT users, and at the same time bring down the charges due to cellular communication.

Of late, the requirement for Quality of Service (QoS) has been gradually shifting more towards the users' experience which is known as Quality of Experience (QoE). QoE indicates the degree of a user's satisfaction from the communication services. Fiedler et~al.~\cite{fiedler2010generic} made an in-depth study to understand the differences between QoE and QoS.  QoE characterizes qualitative aspects, such as a user's perception, experience, and expectations from an application and the network performance. On the other hand, QoS is primarily a quantitative measure to capture the quality aspects as perceived by a user.  Therefore, QoE more appropriate for characterizing the quality of the communication service provided by the operators~\cite{sousa2017survey}. 

In~\cite{alfayly2012qoe}, an intelligent spectrum allocation in D2D communications is used instead of Internet access to enhance QoE. In our previous work~\cite{tetarave2018v}, we proposed an efficient file sharing mechanism in P2P overlay, wherein a selected subset of mobile stations (MSs), known as pilots, are organized into a DHT ring over WiFi. Each pilot works as a super peer and provides file access capability to all other MSs within D2D (or Bluetooth) range of the pilot. 
However, we realized that the random pilot selection process would lead to uneven drainage of batteries from the pilots due to the imbalance in the data retrieval load. Furthermore, it would also affect the connectivity among the associated pilot members.

In this work, we address the two critical issues mentioned above and analyze the effect of the parameters like minimizing communication distance and data overloading in the selection of tier-1 devices or pilots. We propose an efficient heuristic for pilot selection in a DHT-based P2P  to balance the data retrieval load among the associated peers. 
We use a combination of two independent ideas, viz., short-range communications and DHT overlays to create enhanced opportunities for D2D. The first one cuts down Internet usage while the second reduces search latency. We validate the proposed method through experiments and simulations which yield improved QoE for sharing of text, image, audio, or video files.

The main contributions of this work are summarized below:
\begin{enumerate}
\item We introduce a heuristic-based DHT smartphone overlay to enhance QoE of the users. A small number of devices are chosen as pilots from a collection of a large number of smartphones using the P-Median technique. We solved P-Median placement using Lagrangian relaxation method, and studied its performance. 
\item A framework has been evolved for measuring user satisfaction from the applications on resource constraint devices.  
It provides a user centric satisfaction parameters according to their preference for an application to compute QoE.
\item The performance of content sharing mechanism is evaluated for both simple D2D and DHT-D2D overlay with respect QoE on Vienna LTE-A simulator. 
The results show that the proposed heuristics enhance the QoE by about $25\%$ when compared to the existing methods.
\end{enumerate}

The rest of the paper is organized as follows. Section~\ref{RelWork} presents the research related to an integrated communication environment where D2D connectivity is co-operational with cellular networks.
Section~\ref{v-chord} gives a brief overview of DHT organization of mobile devices. The details of the optimized pilot selection process, the necessity of this selection strategy both from theoretical and experimental perspectives
are presented in Section~\ref{effPilotSelection}. Section~\ref{QuaAsse} deals with the quality assessment framework and its analysis with different number of parameters. Section~\ref{sim} presents simulation and analysis of the proposed strategy of integrating D2D with LTE-A communication.  Finally, we conclude our work and provide some future research directions in Section~\ref{con}.

\section{Related Work}
\label{RelWork}
A fair amount of research has been carried out on the strategies that minimize the use of broadband access in GSM/LTE-A networks for retrieving shared files.  McNamara et~al.~\cite{mcnamara2008creating} proposed an inexpensive mobile P2P file sharing environment called JBPeer using 3G and Bluetooth. JBPeer uses 3G for transmission of control data and Bluetooth for communicating shared data. In~\cite{camps2013device} Camps-Murr et~al. gave an extensive overview concerning the ability of WiFi direct devices to discover and establish connections with one another. The main focus of their research was for designing enhanced power saving protocols.   
Doppler et~al.~\cite{doppler2009device} proposed D2D connections as an underlying layer of the LTE-A network. Their focus is to limit the interference of D2D connections in the 3GPP LTE-A network. The target of the approach is to allow the setup and the management of D2D sessions in an overall communication framework of the LTE-A  network.  In~\cite{corson2010toward}, the authors presented a proximity-aware communication system, named Aura-net.  The D2D communication in Aura-net takes place by proximity region without cellular network services.  

Vicinity-based inexpensive and low overhead data retrieval would be highly preferable by smartphone users. Liu et al.~\cite{liu2015device} proposed an SDN (Software Defined Network) based multi-tier LTE-A network to access the Internet from one hop (D2D) devices without cellular communication for improved QoE. 
Zhu et al.~\cite{zhu2015qoe} discussed a QoE-based scheduling scheme to provide an efficient cellular radio resource allocation for video streaming through D2D communication. It constrains the stall events for each stream to maximize the average quality of video time streaming. 
Alfayly et~al.~\cite{alfayly2012qoe} evaluated the performance of the scheduling algorithms in terms of QoE for VoIP applications over the LTE network.  

Unfortunately, the majority of the aforementioned mechanisms to calculate QoE are complex and not suitable for smart devices with limited resource capacity, such as mobile phones or PDAs. It is, therefore, very important to obtain a light weight QoE computation measurement for resource constraint devices. Such an efficient user satisfaction framework using a DHT-based P2P overlay is described later in Section~\ref{QuaAsse}.

In the next section, we illustrate the basic organization and file retrieval process of  DHT overlay using a recently proposed file sharing mechanism called V-Chord~\cite{tetarave2018v} to understand the different QoE factors involving in DHT-based overlays. Through V-Chord, we  analyze the importance of optimized overlay connectivity, and thereby enhancing QoE.

\section{V-Chord: An efficient file sharing on LTE/GSM network}
\label{v-chord}
V-Chord consists of MSs or devices with multiple communication interfaces such as Bluetooth, D2D, IR, WiFi and cellular networks (GSM/ 3G/ 4G/ LTE), as depicted in Figure~\ref{fig:threeCom}. Some MSs might not have enabled data service over their cellular interfaces, but can avail D2D (or Bluetooth) data services through low range wireless interfaces within a vicinity. 
MSs which have enabled broadband data service  or WiFi connection are eligible for the role of pilots.  MSs within vicinity either communicate through WiFi connected pilots (vicinity 1 and 2) without using Internet access or through Internet gateway (vicinity 1 and 4) as shown in the Figure~\ref{fig:threeCom}.
Each MS is assigned with a unique ID  to map into the V-Chord ring using DHT over GSM/LTE underlay infrastructure.  In this, each ID is assigned with an $m$-bit DHT overlay ID. The $b$ most significant bits out of $m$ bits specify the corresponding $eNodeB$. The next $p$ bits specify the associated pilot and the remaining $h$ bits define the MS-ID under the pilot.
\begin{figure*}[!t]
    \centering
    \includegraphics[width=5.0in]{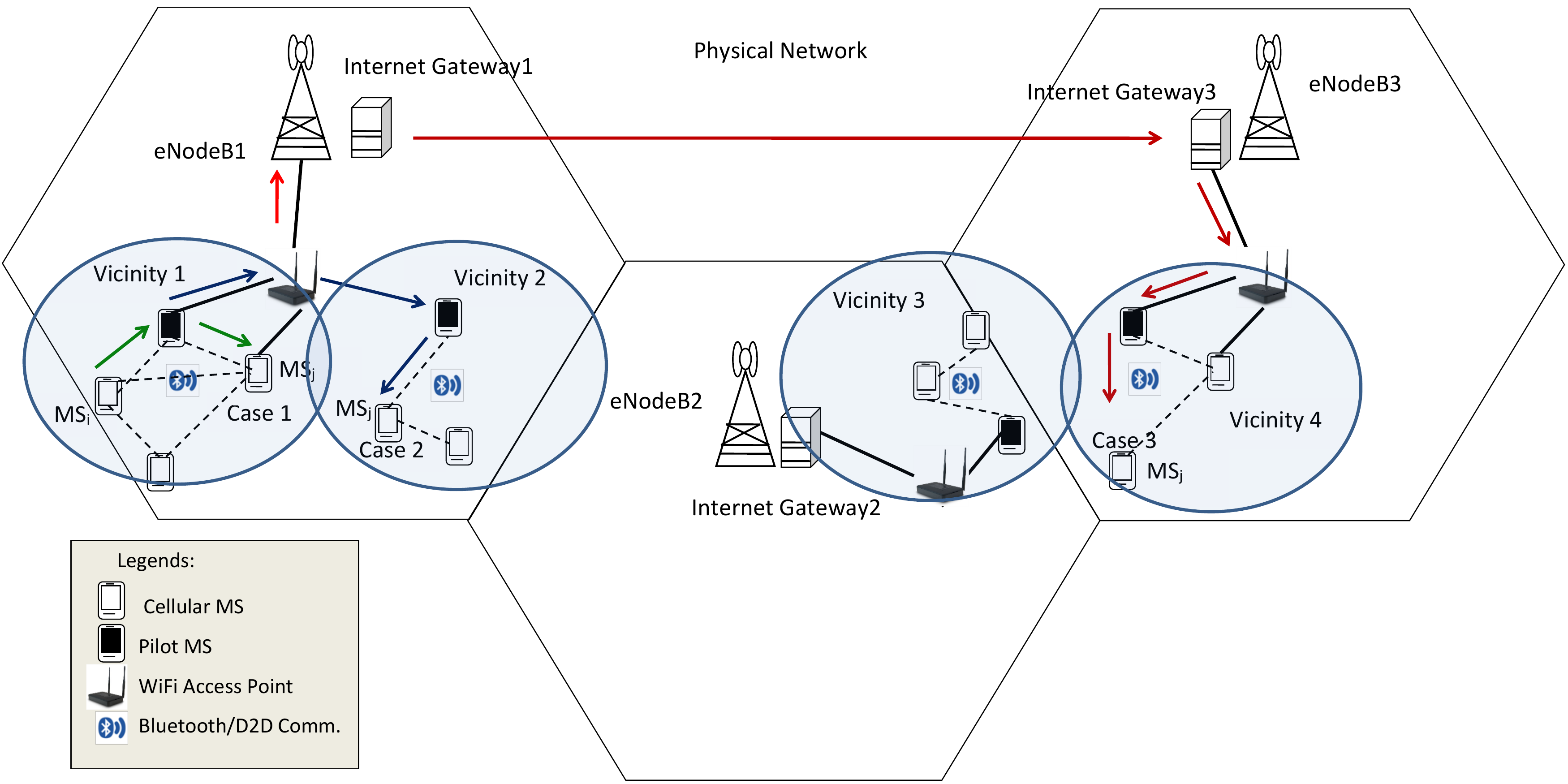}
    \caption{An Overlay with underlay connection in mobile P2P.}
    \label{fig:threeCom}
\end{figure*}

Each $MS_i$ can connect to its neighbor through Bluetooth (or other low range wireless interfaces) while $Pilot_i$ can communicate directly with its $eNodeB_i$.  The user of an MS may be motivated to volunteer the device for the role of a pilot through inducements such as discounted data tariff rate, high data transfer speed, etc. Each eNodeB maintains a list of selected pilots and their corresponding pilot overlay ID ($PID$) for the V-Chord ring. Table~\ref{OverNotation} summarizes the notations used in V-Chord paper.
\begin{table}

\caption{Nomenclature}
\centering  
\label{OverNotation}  
\begin{tabular}{|l|l|}

\hline \multicolumn{1}{|c|} {Notation} & \multicolumn{1}{|c|}{Definition} \\

\hline   $MS_i/MSID_i$ &    Mobile station $i$ and its  overlay identity \\ 
\hline
$Pilot_i/PID_i$ &  Pilot $i$ and its overlay identity \\ 
\hline
$Vicinity_i/VID_i$ &   Vicinity $i$ and its overlay identity \\ 
\hline
$eNodeB_i/BID_i$ &   eNodeB $i$ and its overlay identity \\ 
\hline
$P$ &   Number of pilots \\
\hline
$PID$ &   A shared file key \\
\hline
$US_i/i$ &   Value of user satisfaction $i$ and \\
                 &   rank of its preference  \\
\hline
\end{tabular} 
\end{table}
\subsection{File Retrieval}
The procedure retrieving a file using D2D communication without DHT overlay is presented in {\em case~0} of Figure~\ref{fig:DHTNonDHT}. In this case, an ($MS_i$) initiates the registration process by sending a request, containing its ID, status and traffic type to $eNodeB_i$.
The $eNodeB_i$ executes Algorithm~\emph{Algo-A} to allocate resources to $MS_i$ for communicating with another mobile station $MS_j$. Each eNodeB allocates suitable pairing resources for establishing communication between the MSs.

The retrieval of a file using D2D communication with DHT overlay is illustrated in Figure~\ref{fig:DHTNonDHT}.  The approaches to handle the different possible cases are described below.
\begin{enumerate}
\item[Case 1]  ($MS_i$ and $MS_j$ in the vicinity of a single pilot) 

In this case, $MS_i$ initiates the process by sending a search request with the ID, $FID$ of the file to $Pilot_i$. 
$Pilot_i$ locates $MS_j$ ($MSID_j$) having the ID with nearest prefix match for $FID$  using algorithm~{\em Algo-B}. The pilot forwards the request to $MS_j$ for transferring the file to $MS_i$. This communication can be carried out either on the Bluetooth link or the D2D link.

\item[Case 2] ($MS_i$ and $MS_j$  are in vicinity of a WiFi connected pilots) 

In this case, on receiving a request from $MS_i$, $Pilot_i$ uses algorithm \emph{Algo-C} for forwarding the request to the nearest prefix match for $Pilot_j$ ($PID_j$) within the WiFi range.
$Pilot_j$ runs algorithm \emph{Algo-D} for a look-up in its table for the IDs of MSs having the nearest prefix match for $FID$ (of the requested file) along with their status information and availability. If found, $Pilot_j$ fetches the data from $MS_j$ and forwards it to $Pilot_i$. After receiving the file, $Pilot_i$ transmits it to the requested station $MS_i$.

\item[Case 3] ($MS_i$ and $MS_j$ under inter region pilots)

If file requested by $MS_i$ is not available within the WiFi range of $Pilot_i$, it executes algorithm  \emph{Algo-E} and forwards the request to its corresponding $eNodeB_i$. In response to the request from $Pilot_i$, $eNodeB_i$ executes algorithm \emph{Algo-F} to determine an appropriate $eNodeB_j$. If an $MS_j$ ($MSID_j$) has closest prefix match with ID available in the neighborhood of $Pilot_j$ within $eNodeB_j$'s range, then it forwards the data to the requested $MS_i$.
\end{enumerate}
\begin{figure}[!t]
    \centering
    \includegraphics[width=5.0in]{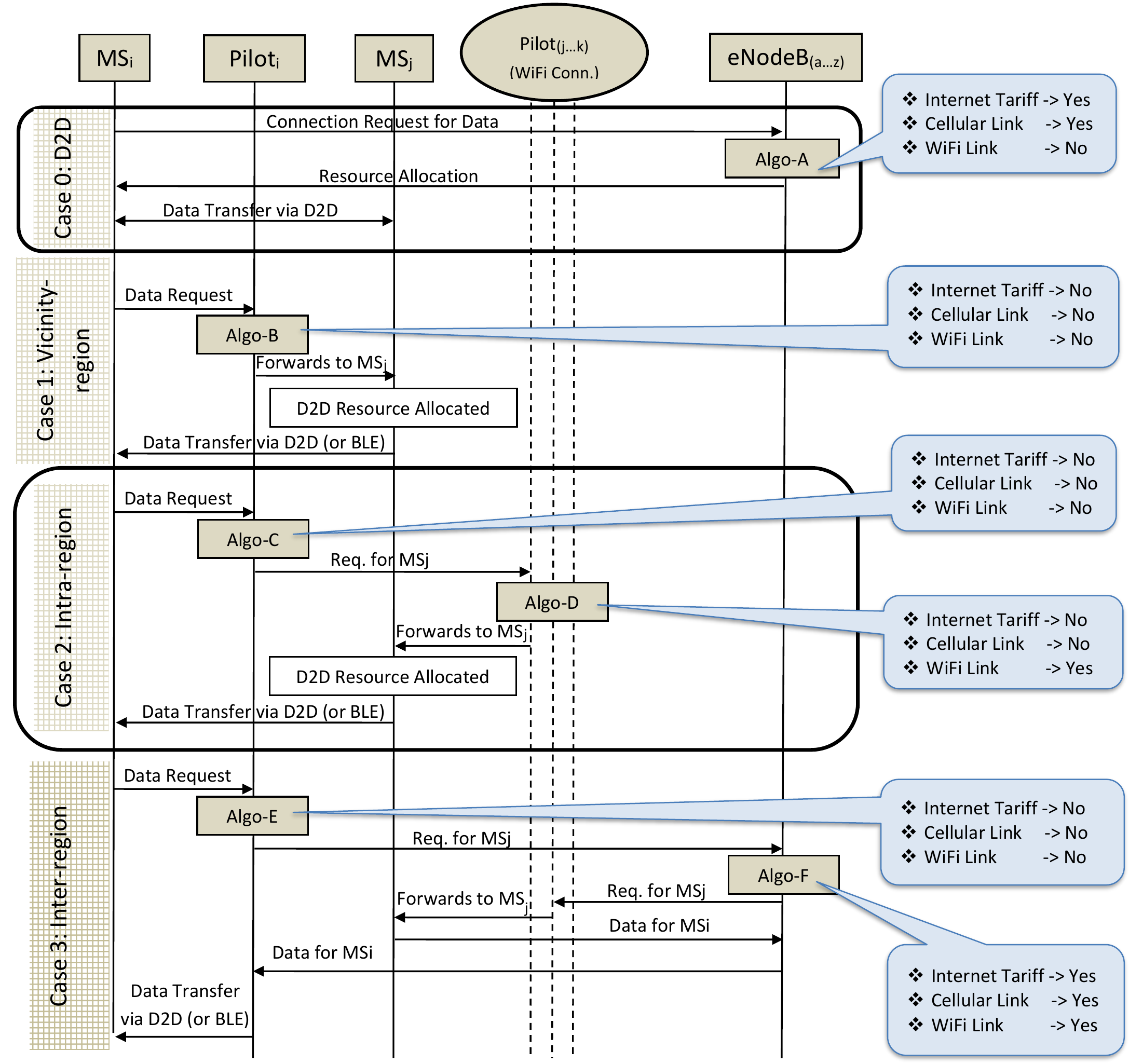}
    \caption{DHT and non-DHT communications under LTE-A network.}
    \label{fig:DHTNonDHT}
\end{figure}

It can be observed that in any of the three cases discussed above, most of the DHT super-peer based overlays would suffer from load balancing and peer association issues. To address these issues, there is a call for an efficient heuristic mechanism for assigning pilots or super-peers in  DHT-based overlays, specifically for resource sharing applications. In the next section, we propose an efficient pilot selection strategy using P-median approach and then, we explain our user satisfaction framework to analyze QoE in such overlays. 
\section{Efficient Pilot Selection Process} 
\label{effPilotSelection}
The most critical step in the realization of efficient D2D  communication is the selection process for the pilot MSs.  Algorithm~\ref{optPilot:1} employs the idea of facility allocation (P-Median) problem~\cite{church1974maximal} for the same. During the setup phase, each  $eNodeB$ executes $Assign(Pilot\_No)$. Subsequently, the pilots who exceed their capacities for serving requests for data in its neighborhood execute  Algorithm~\ref{optPilot:1}. The algorithm reassigns the MSs or the pilots and re-balances the load in the vicinity. 
For convenience in describing the algorithms, we have used a few notations as in Table~\ref{GSMNotations}. 
\begin{table}[!t]
\centering
\caption{Notations used to analyze and evaluate optimized pilot assignment.}
\label{GSMNotations}       
\begin{tabular}{|p{2.7in}|l|}

\hline \multicolumn{1}{|c|} {Description} & \multicolumn{1}{|c|}{Notation (Value)} \\
\hline
Size of shared data at each MS ($MS_i$).   &  $d_{i}$ \\
\hline
Size of shared data at each pilot ($Pilot_j$).   &  $d_{j}$ \\ 
\hline
Number of MSs under each region & $m$ \\
\hline
Distance from $Pilot_j$ to $MS_i$ & $h_{ij}$ \\ 
\hline
Number of eligible pilots under each WiFi connected Location & $e$ \\
\hline
Number of required pilots under a region & $P_{num}$ \\ 
\hline
Maximum data serving capacity for a pilot   & = $P_{cap}$ \\ 
\hline
Initial penalties of Lagrangian relaxation  for data capacity $P_{cap}$ &  $\lambda'_j$ (=1/$d_{j}$) \\
\hline
Initial penalties of Lagrangian relaxation  for $MS$ association with pilots &  $\mu'_i$ (= m) \\
\hline
A constant value on $k^{th}$ iteration for sub-gradient method &  $A^{(k)}$ \\
\hline
Decision variable for pilot assignment &  $Z_j$ \\
\hline
Decision variable for MS to pilot assignment &  $Y_{ij}$ \\
\hline
\end{tabular} 
\end{table}
\subsection{Pilot Assignment Process}
We introduce two decision variables ($Z_j$ and $Y_{ij}$) for assignment of required number of pilots ($P$) that may cover the region of participating $DHT$ $MSs$.
\begin{gather*}
Z_j=\left\{\begin{array}{lll}
                1 & \textnormal{if $Pilot_j$ is assigned}  \\
                0 & \textnormal{otherwise}
            \end{array}\right.
\quad
Y_{ij}=\left\{\begin{array}{lll}
                1 & \textnormal{if $MS_i$ is associated by $Pilot_j$} \\
                0 & \textnormal{otherwise}
            \end{array}\right.
\end{gather*}

The objective is to minimize the distance and data overload on each pilot. Thus, the objective can be formulated as
in Eqn.~\ref{herustic}
\begin{equation}
\mini_{\mathbf{(data, distance)}} \sum_j \sum_i d_ih_{ij}Y_{ij} \hspace{3ex}. 
\label{herustic}
\end{equation}

Subject to the following constraints.
\begin{eqnarray}
 \sum_j Z_j = P \hspace{6ex} \\  
 \sum_j Y_{ij} = 1 \hspace{2ex} \forall (MS_i) \hspace{6ex}\\ 
 \sum_i (d_i + d_j)Y_{ij} \leq P_{cap} \hspace{2ex} \forall (Pilot_j) \hspace{6ex} \\ 
 Y_{ij} = 0, 1 \hspace{2ex} \forall (Pilot_j) \hspace{6ex} \\ 
 Z_{j} = 0, 1 \hspace{2ex} \forall (MS_i,Pilot_j) \hspace{6ex} \\ 
 d_{i} \geq 0, \hspace{1ex} d_j \geq 0 \hspace{1ex} and \hspace{1ex} h_{ij} \geq 0 \hspace{6ex} 
\end{eqnarray}

Each constraint reflects a specific purpose. For example, 
Constraint (2) ensures that exactly $P$ pilots should be assigned.
Constraint (3) ensure that each MS is connected to exactly one pilot.
Constraint (4) ensures that the load (shared data) on each pilot is not higher than the pre-defined data serving capacity ($P_{cap}$). 
Constraints (5) and (6) are for appropriate settings of the decision variables. For example, $Y_{ij}$ is set to 1 or 0 depending on whether $MS_i$ is served by $Pilot_j$ or not.  
Similarly,  variable $Z_j$ is set to 0 or 1 depending on whether an MS$_j$ is selected as a pilot or not. 
Constraint (7) states that fixed variables (i.e., $d_i$, $d_j$ and $h_{ij}$) are positive integers.

In the formulation of constraints, variables $Z_j$ and $Y_{ij}$ are restricted to be either 0 or 1. So,  it is difficult to solve the objective function (Eqn.~\ref{herustic}) applying the linear programming technique. Therefore, we have used Lagrangian relaxation (LR) which removes some complicated constraints by adding a penalty to the objective function for each removal.

Let $\lambda=\{ \lambda_j, j \in [1,P]\}$ and $\mu=\{ \mu_i, i \in [1,m]\}$  are two sets of Lagrangian multiplier penalties associated with constraints (4) and (5) respectively. Thus, the corresponding Lagrangian function becomes  
\begin{eqnarray}
\begin{split} 
M(\lambda, \mu)= \sum_j \sum_i d_ih_{ij}Y_{ij} + 
              \sum_j \lambda_j (\sum_i (d_i + d_j)Y_{ij} -  P_{cap}) \\ + \sum_i \mu_i (1 - \sum_j Y_{ij}).
\label{relaxed}
\end{split}
\end{eqnarray}
\begin{figure}[!t]
    \centering
    \includegraphics[width=2.5in]{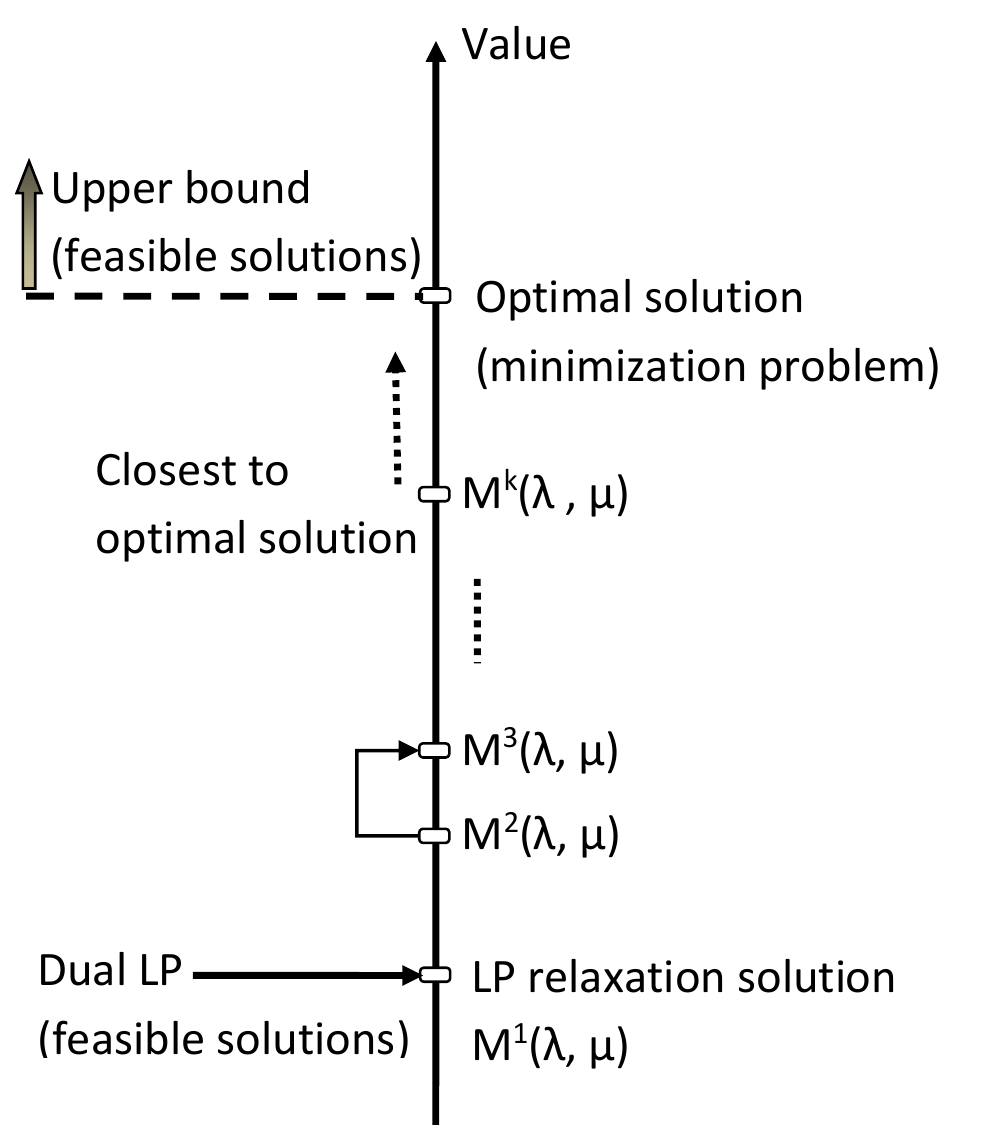}
    \caption{Lagrangian dual problem with sub-gradient solution.}
    \label{fig:LagSubgradient}
\end{figure}
Now, the Lagrangian dual objective function $M(\lambda, \mu)$ would be
\begin{equation}
 \max M(\lambda, \mu).
 \label{langMax}
\end{equation}

 Moreover, as $M(\lambda, \mu)$ is linear in $\lambda$
and $\mu$, the maximum of $M(\lambda, \mu)$ always exists,  and to obtain the best possible solution, we use sub-gradient optimization~\cite{polyak1969minimization}. 
Sub-gradient optimization is an iterative procedure in which  a new approximation for $k+1$ of the Lagrangian multipliers is chosen as
\begin{equation}
 M^{(k+1)}(\lambda, \mu)= \max \{ M(\lambda^1, \mu^1), M(\lambda^2, \mu^2), ..., M(\lambda^{k-1}, \mu^{k-1}), M(\lambda^k, \mu^k) \}.
\end{equation}

The solution of Eqn.~\ref{herustic} is an upper bound for Eqn.~\ref{langMax}.  
 $M^{(k+1)}(\lambda, \mu)$ would be optimal or near-optimal solution of Eqn.~\ref{herustic}, as shown in  Figure~\ref{fig:LagSubgradient}, for $k \rightarrow \infty$.

The overall optimization process is captured by Algorithm~\ref{optPilot:1}.
The step-size value $t^{(k)}$ is updated by Polyak-type step-size rule~\cite{polyak1969minimization}. Polyak proved that if $\epsilon < A^{(k)} \leq D$ for any fixed $\epsilon > 0$, the sub-gradient method is guaranteed to converge, where $D$ depends on the problem definition. 

\RestyleAlgo{boxruled}
\begin{algorithm}  [!th]
\normalsize
\caption{$Assign(P)$} \label{optPilot:1} 

\textbf{Initialize} $A^{(1)} = D $, $\lambda = \frac{1}{\sum_{j=1}^e d_j}$, $\mu= m$, 
$Z_j=1$ and $Y_{ij}=1$,
where $j \in [1, e]$ and $i \in (1, m)$  

\ForEach{ q = 1 : e} 
{  
$V_q = \sum_i min \{ 0, [d_ih_{iq} + \lambda_q(d_i + d_q) - \mu_i] \} $ 
}

\textbf{Select} P Pilots from $V[1,...,e]$, where $V_q$ is minimum 

\ForEach{ r = 1 : $P$} 
{
	\ForEach{ i = 1 : m}
	{
   		\If {($Z_r$ =1 and [$d_ih_{ir} + \lambda_r(d_i + d_r) - \mu_i] < 0$)}
   		{ 
        	\textbf{Assign} $MS_i$ to $Pilot_r$ (i.e., $ V_r$)
        }
   }   
}
 
\ForEach{ k = 1 : iteration}
{ 

\textbf{Compute} 
$t^{(k)}$
	 $= \frac{A^{(k)}(0.01M(\lambda, \mu))}{\sum_r^P \{\sum_i^m (d_i + d_r)Y_{ir} - P_{cap} \}^2 + \sum_i \{1 - \sum_r^P Y^{(k)}_{ir} \}^2}$
     

\ForEach{ s = 1 : P}
{
$\lambda^{(k+1)}_s = max \{0, [\lambda^{(k)}_s - t^k(\sum_i^m (d_i + d_s)Y^{(k)}_{is}- P_{cap})] \}$ 
}

\ForEach{ i = 1 : m}
{
$\mu^{(k+1)}_i = max \{0, [\mu^{(k)}_i - t^k (1 - \sum_j^P Y^{(k)}_{ij})] \}$ 
}

\If{ ($ |M(\lambda^{k+1}, \mu^{k+1}) - M(\lambda^{k}, \mu^{k})| \leq \delta$ )}     
	{
      $A = \frac{A}{2} $
      
      k=1
    }

}
\end{algorithm}

\subsection{Requirement for overlay optimization}
In pure LTE-A networks, QoE for file sharing application depends largely on network traffic. User satisfaction decreases as network traffic increases. The performance bottleneck due to data traffic can be avoided using D2D communication when supported by DHT-based overlays. The pilots handle much of the file sharing loads over short range wireless networks such as Bluetooth, WiFi. It reduces Internet usage on LTE-A networks. Therefore, user satisfaction is expected to increase substantially. However, without proper pilot selection procedure, the load distribution may be unbalanced, which may turn some of the selected pilots inactive due to fast drain out of energy.
\subsubsection{Motivating example}
We examine a scenario as shown in Figure~\ref{fig:OptSelecEx} to understand the energy consumption of pilots under V-Chord. In this case, eNodeB is situated at point (0, 0) with a transmission range of 250-meter. All D2D devices are placed under eNodeB with a transmission range of 20-meter. We assume that all the pilots have a uniform transmission range of $r = 20$-meter. 
An MS at position $(x_i, y_i)$ is a member of a group of MSs connected to a pilot at $(x,y)$ if ${(x_i - x)}^2 + {(y_i - y)}^2 \le r^2$. 

We evaluate different cases for selecting the pilots while forming a DHT overlay. In random selection (case 1) as shown in Figure~\ref{fig:OptSelecEx}, three $MS IDs$ (7, 4, 1) are selected randomly as pilots.  Pilot ID $PID_1$ has four ($MS IDs$ $MS_2$, $MS_3$, $MS_5$, and $MS_{10}$) D2D MS members. $PID_2$ has no D2D $MS$ within its range, while $PID_3$ has three ($MS_6$, $MS_8$ and $MS_9$) $MSs$. So, pilots $PID_1$, $PID_2$ and $PID_3$ serve data to associated MSs up to 1675, 250, and 575 MB respectively.    
\begin{figure*}[!t]
    \centering
    \includegraphics[width=4.5in]{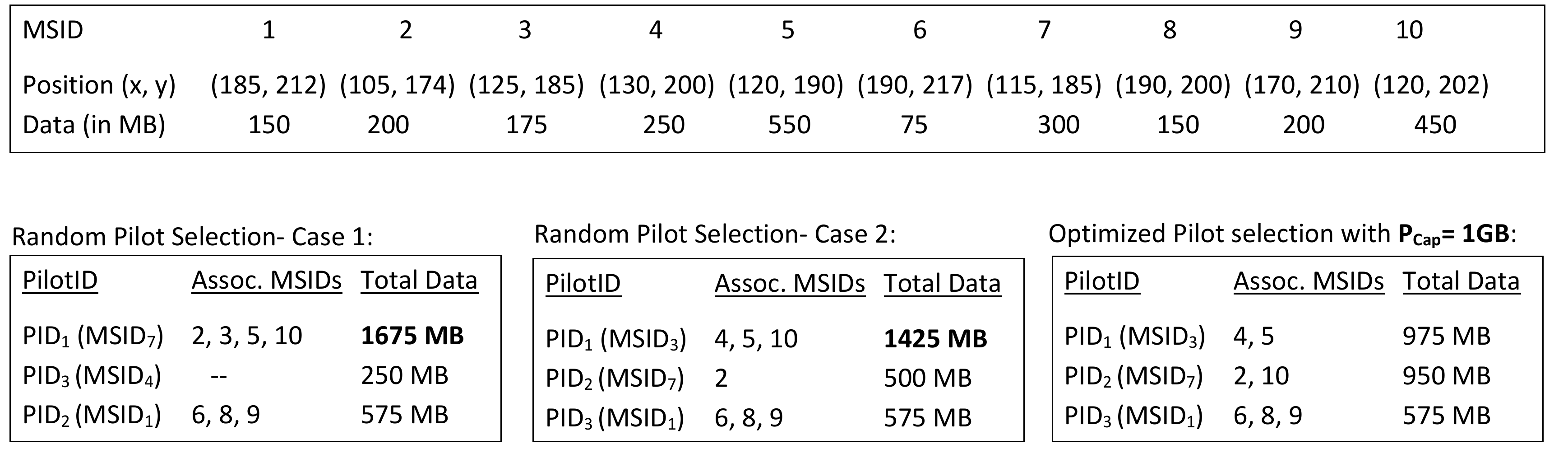}
    \caption{An example of optimized pilot selection where MSs = 10, ISD = 250-meter, D2D range = 20-meter.}
    \label{fig:OptSelecEx}
\end{figure*}

Figure~\ref{fig:OptSelecEx} illustrates another random selection of pilots where $MS_3$ is selected before $MS_7$. In this case, pilots $PID_1$, $PID_2$ and $PID_3$ serve data to associated $MSs$ up to 1425, 500 and 575 MB respectively. In both cases, the selection of pilots could not balance their data serving capacity. Unbalanced load (for data serving) distribution would cause higher energy drain outs from some pilots. 

The most simple solution would be to put a cap on data serving capabilities of each pilot. We experiment with an optimal pilot selection mechanism with $P_{cap}$ (= 1 GB), the load distribution becomes balanced as shown in Figure~\ref{fig:OptSelecEx}.
However, it can be viewed as per pilot optimized solution. However, it is highly inadequate as many MSs may not find a pilot nearby. 
Therefore, an optimal or near optimal placement of the pilots is required to provide high level of user satisfaction.
\subsection{Analyzing the proposed optimization technique}

The proposed pilot selection mechanism is 
simplified using Lagrangian Relaxation and analyzed its rate of convergence towards an optimal solution 
using MATLAB R2016a on a PC with Intel Core i5 processor with 8GB of main memory in Figure~\ref{fig:Opt-a}. The convergence rate of our simplified mechanism significantly depends on the initial assignment of $A^{(k)}$, which would be chosen according to Pearson's correlation coefficient. Pearson's correlation coefficient ($\sqrt{R^2}$)  is found to be 0.965 from (Figure~\ref{fig:Opt-b}). Further, the relation between the number of MSs ($m$) and $A^{(k)}$ is derived to be,  ($A^{(k)}= 0.017m - 2.9412$).

 %
\begin{figure}[!t]
    \centering
    \includegraphics[width=2.5in]{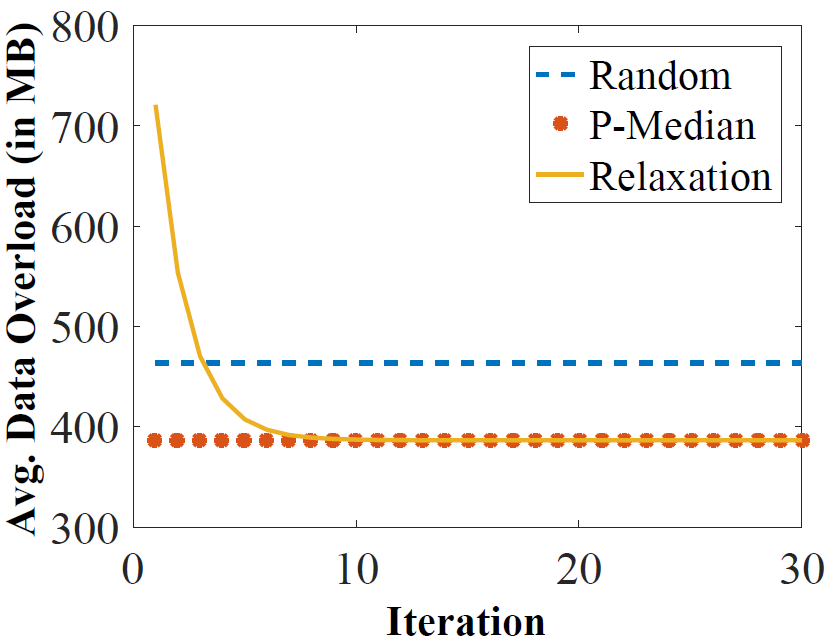}
    \caption{Convergence of simplified mechanism (using Lagrangian Relaxation) towards the optimized solution.}
    \label{fig:Opt-a}
\end{figure}
\begin{figure}[!t]
    \centering
    \includegraphics[width=3.0in]{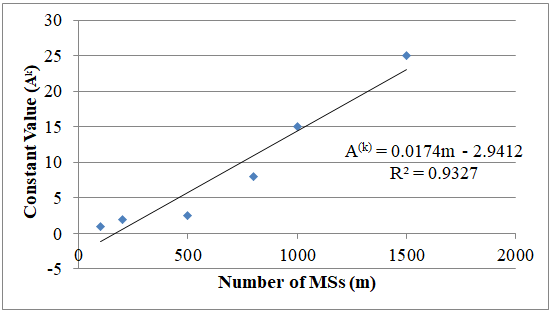}
    \caption{Pearson's correlation coefficient between MS and $A^{(k)}$.}
    \label{fig:Opt-b}
\end{figure}
\subsubsection{Evaluating on a real samrtphone}
We evaluated the performance of the proposed optimized pilot selection mechanism using an Android App based smartphone.
The App has been built on the Ionic framework using Cordova plug-in to run on a smartphone having android OS ~\cite{latif2016cross}.
The implementation is configured to optimize a vicinity of 20 MSs with sharing the capacity of (up to) 500 MB (for each MS). The optimized process performs for 250 times on
Motorola Nexus 6 smartphone having OS (Android 7.1.1), CPU (Quad-core 2.7 GHz Krait 450), internal memory (3 GB), and battery (3220 mAh).  The results are displayed in Figure~\ref{fig:CPUTotal-a} and ~\ref{fig:CPUTotal-b}. It is observed that the App requires 1.3 MB memory, 18 seconds CPU time with the energy consumption of 1 mAh.

   \begin{figure}[!ht]
   \centering
     \subfloat[App information regarding memory and battery consumption. \label{fig:CPUTotal-a}]{%
       \includegraphics[width=0.35\textwidth]{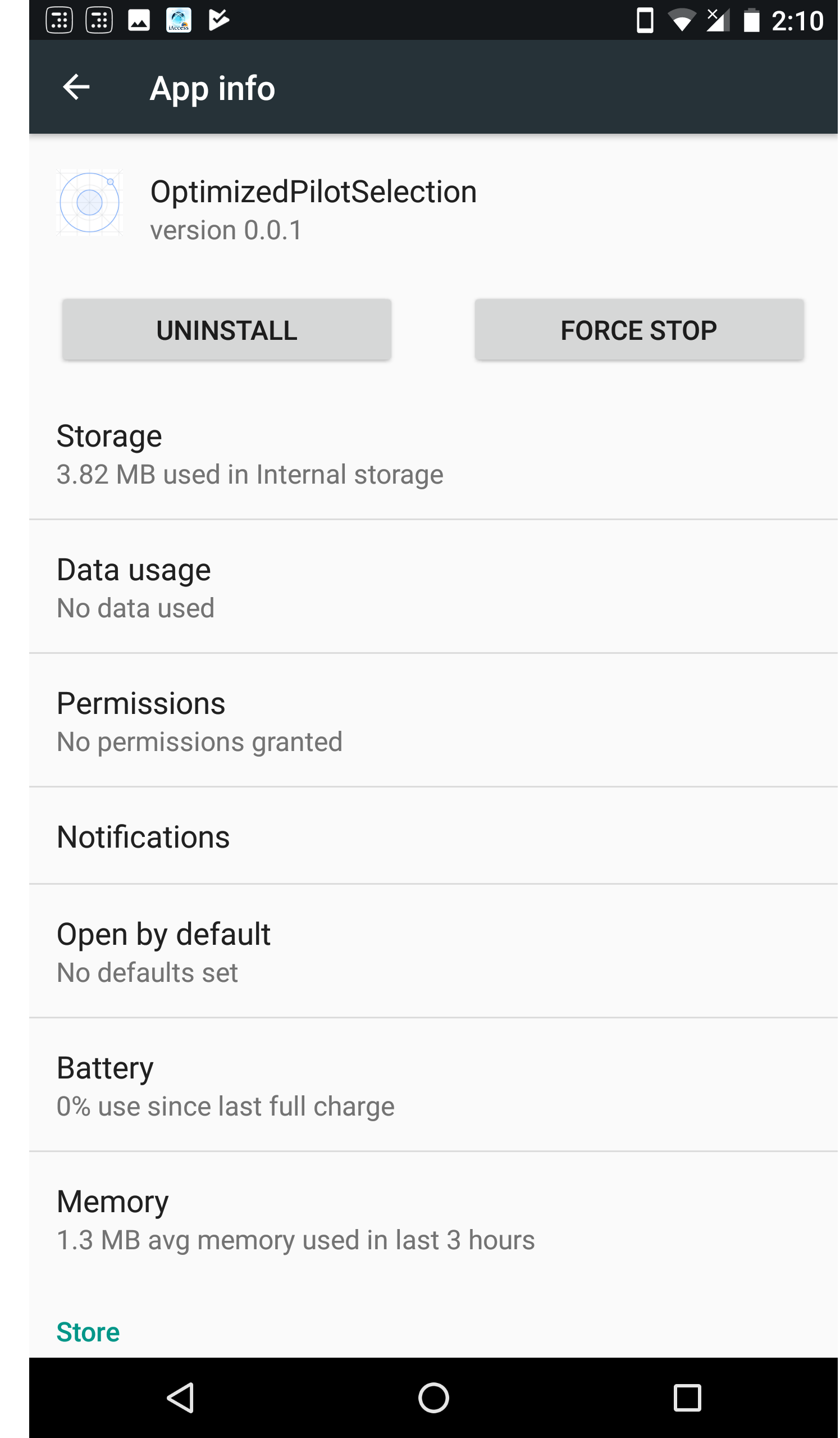}
     }
     \subfloat[Computational power to perform pilot optimization process up to 250 times. \label{fig:CPUTotal-b}]{%
       \includegraphics[width=0.35\textwidth]{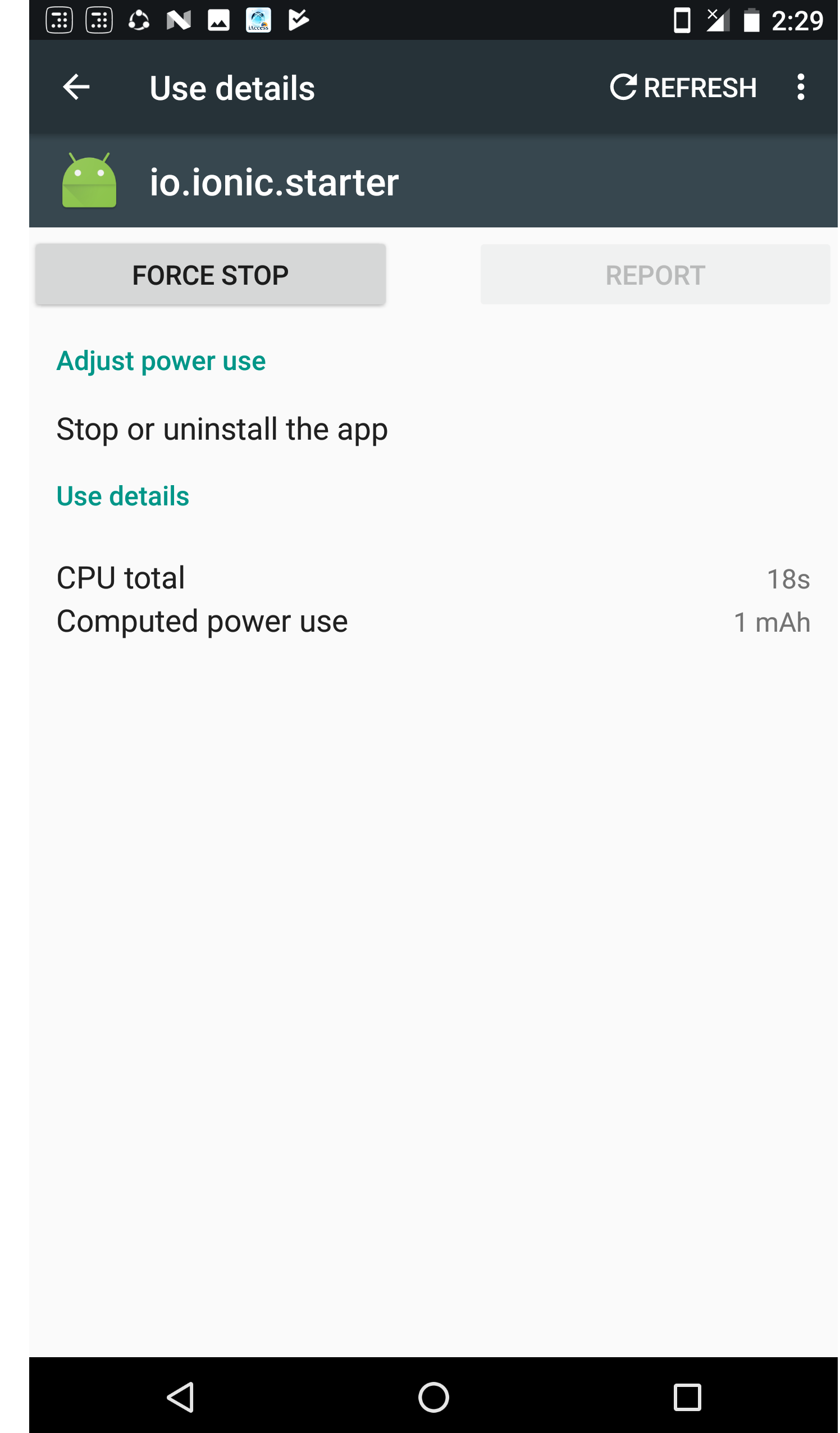}
     }
     \caption{Optimization process of pilots on a smartphone.}
     \label{fig:CPUTotal}
   \end{figure}

\subsection{Discussion}
According to our problem, $D$ would be [$0.017m - 2.9412$] (derived from Figure~\ref{fig:Opt-b}). 
To compute $t^{(1)}$  (in Algorithm~\ref{optPilot:1}), $A^{(1)}$ is initialized with $D$ and to compute $M(\lambda, \mu)$, $\lambda$ and $\mu$ are initialized with $\frac{1}{\sum_{j=1}^e d_j}$ and $m$ respectively. 
If the solution does not improve significantly, $A^{(k)}$ needs to be reduced (e.g., $A^{(k+1)} = 0.5A^{(k)}$). Thus, $A^{(k)} \geq A^{(k+1)}$. 

The overall time complexity of the proposed pilot assignment is $O(2kmP)$. However, we observe from our simulation that the algorithm converges with small $k$ (e.g., 8), so the complexity would be $O(mP)$.
\section{Quality Assessment}
\label{QuaAsse}
QoE is referred to as "Degree of the delight in using a service"~\cite{brunnstrom2013qualinet}. The perceived QoE may change in an end-to-end communication service, due to the dynamic temporal and spatial conditions (remaining battery energy, interest in the use of Internet service, a distance between smartphones, etc.) of mobile users. 
The communication requirements in QoE can be influenced by the content, the network, the device, the application, the context of use, the user expectations and the goals.

According to Sousa et al.~\cite{sousa2017survey} the estimation of QoE is influenced by several factors related to system, service, application, the user's preferences or context whose actual state or setting may have an influence on the QoE for the user. These factors can be categorized as a human, system, and contents based. A user's preferences are largely influence by attributes such as gender, age, audio-visual capabilities, etc. Sometimes, intellectual capabilities such as motivation, encouragement, spatial and temporal mood, etc., can also be critical and may affect the QoE. QoE also gets affected by the systems through which it is perceived, for example, multimedia streaming that goes through compression and decompression at multiple levels of transmission. Apart from the above two major factors, the other important factor is the external environment. These can be the temporal aspects, like the day of the weak or time of the day, duration of the content and its popularity, and service type, etc.~\cite{sousa2017survey}.

Although, QoE can depend on the aforementioned multiple factors, and therefore, may be measured by various qualitative and quantitative metrics. An application dependent QoE is quantified by a mean opinion score (MOS) value, which is calculated using the degree of subjective satisfaction for the end users of a particular application~\cite{liotou2014quality}.  In cellular networks, smartphone users satisfy with an application such as the effect of application on their smartphone components, tariff, battery consumption using the service such as file sharing. Table \ref{parametersTable} shows a few parameters that can be used to analyze smartphone in such applications.
\begin{table}
 
\caption{User Satisfaction Parameters.}
\centering    
\label{parametersTable} 
\begin{tabular}{|l|l|}
\hline \multicolumn{1}{|c|} {Preference} & \multicolumn{1}{|c|}{Parameter} \\

\hline   
1 &    Number of accessed files without using the Internet \\ 
\hline
2 &  Accessing time for all file chunks\\ 
\hline
3 &   Energy consumption of user’s mobile during file retrieval\\ 
\hline
4 &    Search rank file within vicinity  \\
\hline
5 &    Search file with keyword within vicinity  \\
\hline
6 &    Target file at single (or multi-hop) hop D2D  \\ 
\hline
7 &  Connecting time with target D2D \\ 
\hline
8 &  Associated pilot at single (or multi-hop) hop  \\ 
\hline
9 &  New peer joining time \\ 
\hline
\end{tabular} 
\end{table}

The MOS can be computed from the statistical analysis of the data provided by all users, and QoE of the service can be calculated using MOS values. But, MOS computation becomes complex as the number of application dependent metrics increases. Therefore, the calculation of QoE is challenging. 
A visual solution is to reduce the metric size. For example, Liotou et al.~\cite{liotou2014quality} uses only one parameter (Internet access) to reduce the complexity.

When more parameters are included it is difficult to determine their relative importances while eliminating from the metric to compute QoE. 
So, we decided to propose a light weighted framework to compute user satisfaction providing a preference for parameters. Further, we analyze the framework that how it helps to reduce the complexity.

\subsection{The proposed user satisfaction framework}   
\label{satframe}
This section describes the proposed user satisfaction framework. A user chooses an application based on a perceived satisfaction score for using the application. It is possible to define a priory range of values for a set of parameters which a user would find important. Then a user's feedback can be gathered by seeking a satisfaction score for each parameter from
the corresponding predefined absolute category rating (ACR) scale. ACR maps ratings between Bad and Excellent to numerical values within the range 1 through 5 ~\cite{ITU:xxx}.
However, in this work, we modify the rating values within the range $[-2, +2]$. More specifically, the value for a parameter $i$ is denoted by $US_i\in\{ 2\ (Excellent), 1\ (Good), 0\ (Satisfactory), -1\ (Poor), or -2\ (Bad)\}$. The primary reason behind the proposed modification is to obtain a smooth convex surface with its minima at 0 that helps in the optimization process, and reduces the requirement to compute large values. Next, we define the overall user satisfaction score ($US_{overall}$) as follows: 
\begin{equation}
US_{overall}= \frac{\sum_{i=1}^{\Bbbk } US_i/i}{\sum_{i=1}^{\Bbbk} 1/i}, 
\label{USEq}
\end{equation}
where $\Bbbk$ denotes the number of parameters chosen by the user about
an application. 
The overall satisfaction depends on the chosen parameters (see Table~\ref{parametersTable}) and lies between -2 and 2, i.e., $ -2  \leq US_{overall} \leq  2$. 
We analyze the behavior of Eqn.~\eqref{USEq} with the help of random harmonic distribution~\cite{schmuland2003random}, to determine the acceptance (or rejection) of the application without computing large MOS value for a small and broad set of user satisfaction parameters.

\subsubsection{Small set of user satisfaction parameters}

\begin{prop} \label{pro:1} If value of each  $US_i$ ($ 1 \leq i \leq \lfloor \frac{\Bbbk}{2} \rfloor$) is 2 (-2), $US_{overall}$ QoE score becomes positive (negative). 
\end{prop}

\begin{proof} 
Considering the value of $US_i \in \{2\}$ $\forall (1 \leq i \leq \lfloor \frac{\Bbbk}{2} \rfloor )$.

$\sum_{i=1}^{\Bbbk} \frac{US_i}{i}$
= $\sum_{i=1}^{\lfloor \frac{\Bbbk}{2} \rfloor} \frac{US_i}{i} $ + $\sum_{j= \lceil \frac{\Bbbk+1}{2}\rceil }^{\Bbbk} \frac{US_i}{j} >$ 
$2\sum_{i=1}^{\lfloor \frac{\Bbbk}{2} \rfloor} \frac{1}{i} $ - 2$\sum_{j= \lceil \frac{\Bbbk+1}{2}\rceil }^{\Bbbk} \frac{1}{j} >$ 0.

Therefore, $US_{overall} > $  0. 

Similarly, we can show that $US_{overall} <$ 0, when $US_i$ = -2 $\forall (1 \leq i \leq \lfloor \frac{\Bbbk}{2} \rfloor )$.
\end{proof}

\begin{prop} If value of each  $US_i$ ($1 \leq i \leq \lceil \frac{\Bbbk}{2}  \rceil$) is $ \in \{1, 2\}$  ($ \in \{-1, -2\}$), $US_{overall}$ QoE score becomes positive (negative). 
\end{prop}
\begin{proof}
Considering the value of $US_i \in \{1, 2\}$ $\forall (1 \leq i \leq \lceil \frac{\Bbbk}{2} \rceil )$. 

Case (i): When $\Bbbk$ is even.

$\sum_{i=1}^{\Bbbk} \frac{US_i}{i}$  $> \sum_{i=1}^{\frac{\Bbbk}{2}} \frac{1}{i} $ - 2$\sum_{j= \frac{\Bbbk}{2}+1 }^{\Bbbk} \frac{1}{j} $ = $\sum_{i=1}^{\frac{\Bbbk}{2}} (\frac{1}{i} - \frac{2}{\frac{\Bbbk}{2} + i}) > $ 0.

Case (ii): When $\Bbbk$ is odd.

$\sum_{i=1}^{\Bbbk} \frac{US_i}{i}$  $> \sum_{i=1}^{\frac{\Bbbk+1}{2}} \frac{1}{i} $ - 2$\sum_{j= \frac{\Bbbk + 3}{2} }^{\Bbbk} \frac{1}{j} $ = $\sum_{i=1}^{\frac{\Bbbk - 1}{2}} (\frac{1}{i} - \frac{2}{\frac{\Bbbk+2}{2} + i})$  + $\frac{1}{\frac{\Bbbk + 1}{2}} >$ 0.

Therefore, in both cases $US_{overall} > $ 0.

Similarly, we can show that $US_{overall} < $ 0, when $US_i \in \{-1, -2\} $ $\forall (1 \leq i \leq \lceil \frac{\Bbbk}{2} \rceil )$.

\end{proof}
\subsubsection{Large set of user satisfaction parameters}
\begin{prop} If feedback value of large set of  user satisfaction are equally distributed (not equally distributed), the application should be accepted (difficult to judge). 
\end{prop}

\begin{proof}

\begin{itemize}
\item[ ]Let S = $US_1 + \frac{US_2}{2} + \frac{US_3}{3}+ ...  + \frac{US_i}{i}  + ... $,
\end{itemize}
where $US_i$ are independent random variables, which may or may not be distributed equally.
In this article, we study and analyze equally distributed or unequally distributed of a large set of random variable $US_i$, which helps users to decide whether an application should be accepted.   
Our analysis covers both cases.
According to Kolmogorov three-series theorem~\cite{wiki:xxx},  the series $\sum_{i=1}^{\infty}X_i$ (where, $(X_i)_{i \in N}$ are independent random variables) converges   $\mathbb{R}$ if and only if the following three conditions (C1, C2 and C3) hold for some $A > 0$ :  
\begin{itemize}
\item[(C1)]  $\sum_{i=1}^{\infty} Pr(|X_i| \geq A) $ converges.
\item[(C2)] Let $Y_i := X_i.1_{\{|X_i| \leq A\}}$, then  $\sum_{i=1}^{\infty} E(Y_i)$, the series of expected values of $Y_i$, converges.
\item[(C3)]  $\sum_{i=1}^{\infty} var(Y_i)$ converges. 
\end{itemize}
where, $Pr(X)$, $var(Y)$ and $E(Y)$ are denoted as probability of $X$, variance and expected values of $Y$ respectively.

Case (i): Let user satisfaction values be equally distributed and
$US_i$s are distributed with uniform distribution 
$Pr(US_i=-2)$ 
= $Pr(US_i=-1)$
= $Pr(US_i=0)$
= $Pr(US_i=1)$ 
= $Pr(US_i=2)$ 
=  $\frac{1}{5}$.
Therefore, $E(S)$ = 0 and 
\begin{equation}
E(S^2)= \sum_{i=1}^{\infty} \frac{E(US_i^2)}{i^2} = \frac{10}{5}(\frac{\pi^2}{6}) = \frac{\pi^2}{3}.
\end{equation}
Since $\frac{US_i}{i}$ are random variables, then

\begin{itemize}
\item[(C1)] $\sum_{i=1}^{\infty} Pr( |\frac{US_i}{i}| \geq 3)$ = 0.
\item[(C2)] If $Y_i$ := $\frac{US_i}{i}.1_{\{|\frac{US_i}{i}| \leq 3\}}$ , then
  \begin{itemize}
      \item[ ] $Y_i= \frac{US_i}{i}$ and
    \item[ ] E($Y_i$) = $\sum_{i=1}^{\infty} E(Y_i)$ = $\sum_{i=1}^{\infty} E(\frac{US_i}{i})$ = 0.  
  \end{itemize}

\item[(C3)] $\sum_{i=1}^{\infty} var(Y_i)$ = $\sum_{i=1}^{\infty} var(\frac{US_i}{i})$ 
   \begin{itemize}
   \item[ ] = $\sum_{i=1}^{\infty} \frac{1}{i^2} var(US_i)$
   \item[ ] = $\sum_{i=1}^{\infty} \frac{1}{i^2} ( E(US_i^2) - [E(US_i)]^2)$
   \item[ ] = $\sum_{i=1}^{\infty} \frac{1}{i^2}.2$ =  $\frac{\pi^2}{3}$. 
   \end{itemize}
\end{itemize}

Therefore, all three conditions of Kolmogorov three-series theorem holds.
This implies $\sum_{i=1}^{\infty} \frac{US_i}{i}$  almost surely converge. 
Since, $\sum_{i=1}^{\infty} \frac{1}{i}$ diverges towards $+ \infty$, therefore, $US_{overall} \rightarrow 0$.
It shows that the equally distributed feedback values for different parameters of an application should be acceptable by users.  

Moreover, from ~\cite{schmuland2003random}, we conclude that
\begin{itemize}
\item[(a)] Probability of S having very large value is very small, but it is never 0.
\item[(b)] Distribution of S has full support on the real line, so, there is no theoretical upper bound or lower bound of S. So range of S is $\mathbb{R}$ U $\{-\infty, \infty\} $.   
\end{itemize}

Case (ii): Let user satisfaction values be not equally distributed and the satisfaction values of different parameters are provided with different probabilities.
Suppose,
 $Pr(US_i=-2)$ = $Pr1$, 
 $Pr(US_i=-1)$ = $Pr2$, 
 $Pr(US_i=0)$ = $Pr3$, 
 $Pr(US_i=1)$ = $Pr4$, 
 $Pr(US_i=2)$ = $Pr5$,
where ($Pr1\ +\ Pr2 +\ Pr3\ +\ Pr4\ +\ Pr5$) = 1.
If $A \subseteq [3, \infty)$, then condition (C2) of Kolmogorov three-series does not hold always as 

\begin{itemize}
\item[ ] E($Y_i$) = E($\frac{US_i}{i}$) 
= $(2.Pr5+Pr4-Pr2-2.Pr1)\sum_{i=1}^{\infty} \frac{1}{i}$.
\end{itemize}

Therefore, $E(Y_i)$ diverges if
\begin{itemize}
\item[ ] $(2.Pr5+Pr4-Pr2-2.Pr1) \neq 0$.
\end{itemize}

However, if $A \subseteq (-\infty, 3)$, then  (C1) does not hold as

\begin{itemize}
\item[ ] $\sum_{i=1}^{\infty} Pr(\frac{US_i}{i} < 3) $ = $ \sum_{i=1}^{\infty} 1 \rightarrow $ diverges.
\end{itemize}

Thus S= $\sum_{i=1}^{\infty} \frac{US_i}{i} $ does not converge.
In this case, $US_{overall}$ diverges towards $- \infty$ or $\infty$; therefore, it is difficult to judge the acceptability of the application.
\end{proof}

\section{Simulation setup and results}
\label{sim}
We used Vienna LTE-Advanced open source system level simulator~\cite{ikuno2010system} to verify the proposed user satisfaction framework and optimized pilot selection mechanism. The simulator uses the object-oriented programming paradigm with a modular code structure allowing us to add extra modules with basic input parameters as described in Table~\ref{SimSetup}.

\begin{table}[!t]
\centering
\caption{Input parameters and their values.}
\label{SimSetup}       
\begin{tabular}{|p{2.2in}|l|}

\hline \multicolumn{1}{|c|} {Parameter} & \multicolumn{1}{|c|}{Value} \\
\hline
Inter-site distance (ISD)   &  250-meter\\
\hline
Number of Single Carrier   &  600\\
\hline
Number of Resource Block ($RB_{num}$) & 100 \\
\hline
Total Power $POW_{total}$ & -10 dB W  \\
\hline
Circuit Power Consumption ($CPOW_{con}$) & 0.05 W  \\
\hline
Carrier Frequency ($C_{freq}$) & 2.15 GHz  \\
\hline
Tuning Step ($T$)   & $20/(\pi*1500^2)$ \\
\hline
Path-loss Parameter ($\alpha$) &  3.5 \\
\hline
Number of Pilots ($P_{num}$) &  10 \\
\hline
Number of D2D users ($D2D_{num}$)     &  100 \\
\hline
Number of user satisfaction parameters ($US_{param}$)  &  3, 6, 9 \\
\hline
Vicinity Size ($V_{size}$)  &  $\lceil (1+( D2D_{num} - P_{num})/P_{num})  \rceil$ \\
\hline
Number of Resource Block per D2D user ($RB_{d2dNum}$)  &  $RB_{num}/D2D_{num}$ \\
\hline
Vicinity radius ($V_{radius}$)  &  ISD/10 meters \\
\hline
\end{tabular} 
\end{table}

\subsection{User Satisfaction Parameters}
We implemented a look-up cost model of V-Chord~\cite{tetarave2018v} using the modified D2D simulator to
calculate the MOS values corresponding to user satisfaction parameters
mentioned in Table~\ref{parametersTable}. In this model, the look-up process for DHT file sharing application is categorized into three categories, namely, inter-region, intra-region, and vicinity-region as depicted in Figure~\ref{fig:DHTNonDHT}.

Inter-region look-up uses the Internet or cellular links through an Internet gateway. Access a file located in a different region consumes more battery power in a mobile device. Therefore, a user's objective would be to reduce the inter-region look-ups. The availability of intra-region file chunks allows fast (high bandwidth) access to a file as the WiFi or Bluetooth can be used in short ranges. 
So, the ability to get files without Internet access would be considered as a user satisfaction parameter. 
Though accessing a file from a local neighborhood is cheaper, it takes some time to establish a connection and send file chunks over the Bluetooth link. So, it may be captured by two conflicting satisfaction values (one low and one high).  Similarly, a user may define satisfaction parameters and give MOS feedback values according to the performance of an application as mentioned in Table~\ref{parametersTable} with preference from 1 (top) to 9 (low). 
The value of each parameter is provided through a predefined policy.  
In the simulation, following $US_i$ values are assigned:
\begin{gather*}
US_i=\left\{\begin{array}{lll}
                2 & \textnormal{if 80\% $ < US_i \leq $100\%}  \\
                1 & \textnormal{if 60\% $ < US_i \leq $80\%}
                \\
                0 & \textnormal{if 40\% $ < US_i \leq $60\%}
                \\
                -1 & \textnormal{if 20\% $ < US_i \leq $40\%}
                \\
                -2 & \textnormal{if $US_i < $20\%}
            \end{array}\right.
\end{gather*}          

\subsection{Result Analysis}
The simulation is initialized with 1K shared files among users and 500 new files are added in each iteration. 
DHT-D2D overlay including meta-data implementation, such as V-Chord, keeps information of shared files at pilots and eNodeBs.  The meta-data at a pilot includes all destination MS IDs those have been downloaded (or stored) of a file within its vicinity. An eNodeB includes all responsible overlay pilot IDs those have been retrieved of a file within its region. 
The cache entries can be updated on insertion or retrieval of a new shared file. A similar update of cache is performed for old files or whenever an MS leaves the overlay, which makes an efficient cache size management.

Figure~\ref{fig:UserSatisfection} illustrates that the user satisfaction of randomly positioned D2D users under a single cell in the LTE-A simulator. It shows that each user adopts the top three parameters in their preference set. In Figure~\ref{fig:userSati-a}, each file search is independent of the previous
search. 
The result shows that the D2D users with DHT overlay achieve three times better QoE than that without using DHT. 
Figure~\ref{fig:userSati-b} shows that user satisfaction reaches the near-maximum value with DHT-D2D in which the meta-data of the previously searched file is available at the pilot node. %
   \begin{figure}[!ht]
   \centering
     \subfloat[All search files are new. \label{fig:userSati-a}]{%
       \includegraphics[width=0.4\textwidth]{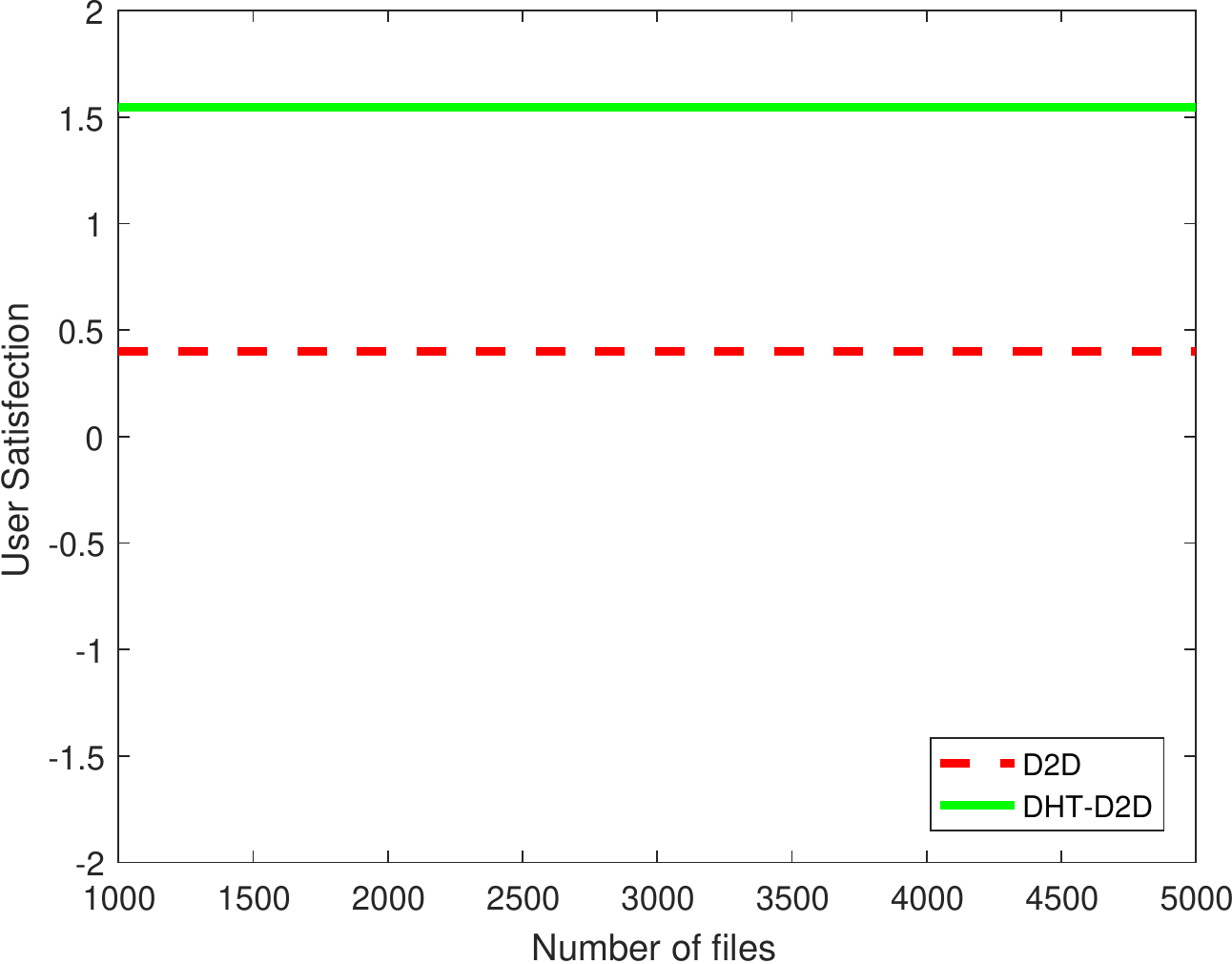}
     }
     \subfloat[Search files are available in the vicinity. \label{fig:userSati-b}]{%
       \includegraphics[width=0.4\textwidth]{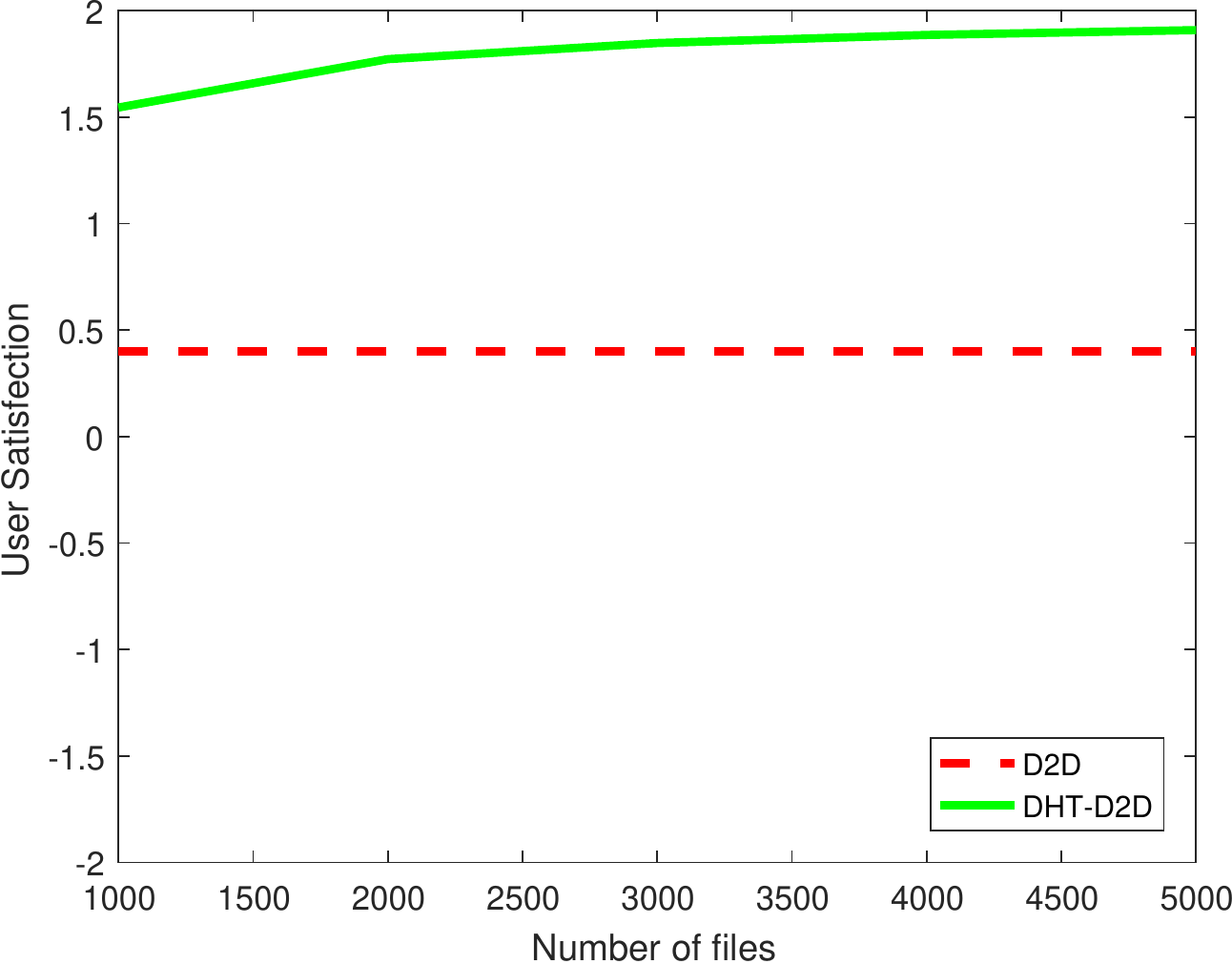}
     }
     \caption{User satisfaction among D2D and DHT-D2D users where 1K new files are added.}
     \label{fig:UserSatisfection}
   \end{figure}

 Figure~\ref{fig:UserParameters} shows the interchanging effect in preference of user satisfaction parameters, which would result in a different QoE.  Figure~\ref{fig:userPar-a} compares the effect in the presence of new files in each iteration where QoE is static. Figure~\ref{fig:userPar-b} shows the improvement of satisfaction due to the available meta-data for the already searched files. Quite clearly, the performance of DHT-D2D is better that employs the use of meta-data of the previous search files.

   \begin{figure}[!ht]
   \centering
     \subfloat[All search files are new. \label{fig:userPar-a}]{%
       \includegraphics[width=0.4\textwidth]{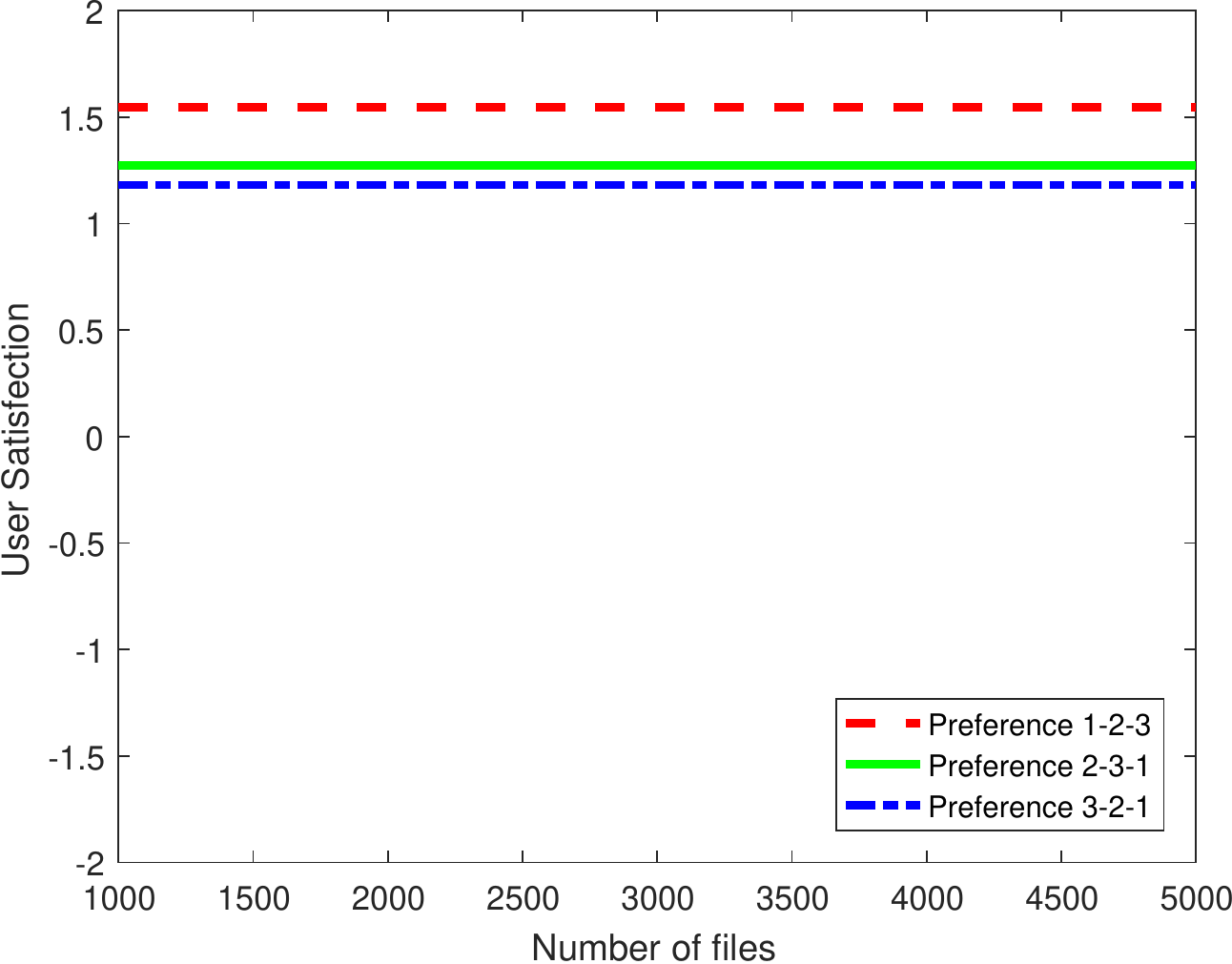}
     }
     \subfloat[Search files are available in the vicinity. \label{fig:userPar-b}]{%
       \includegraphics[width=0.4\textwidth]{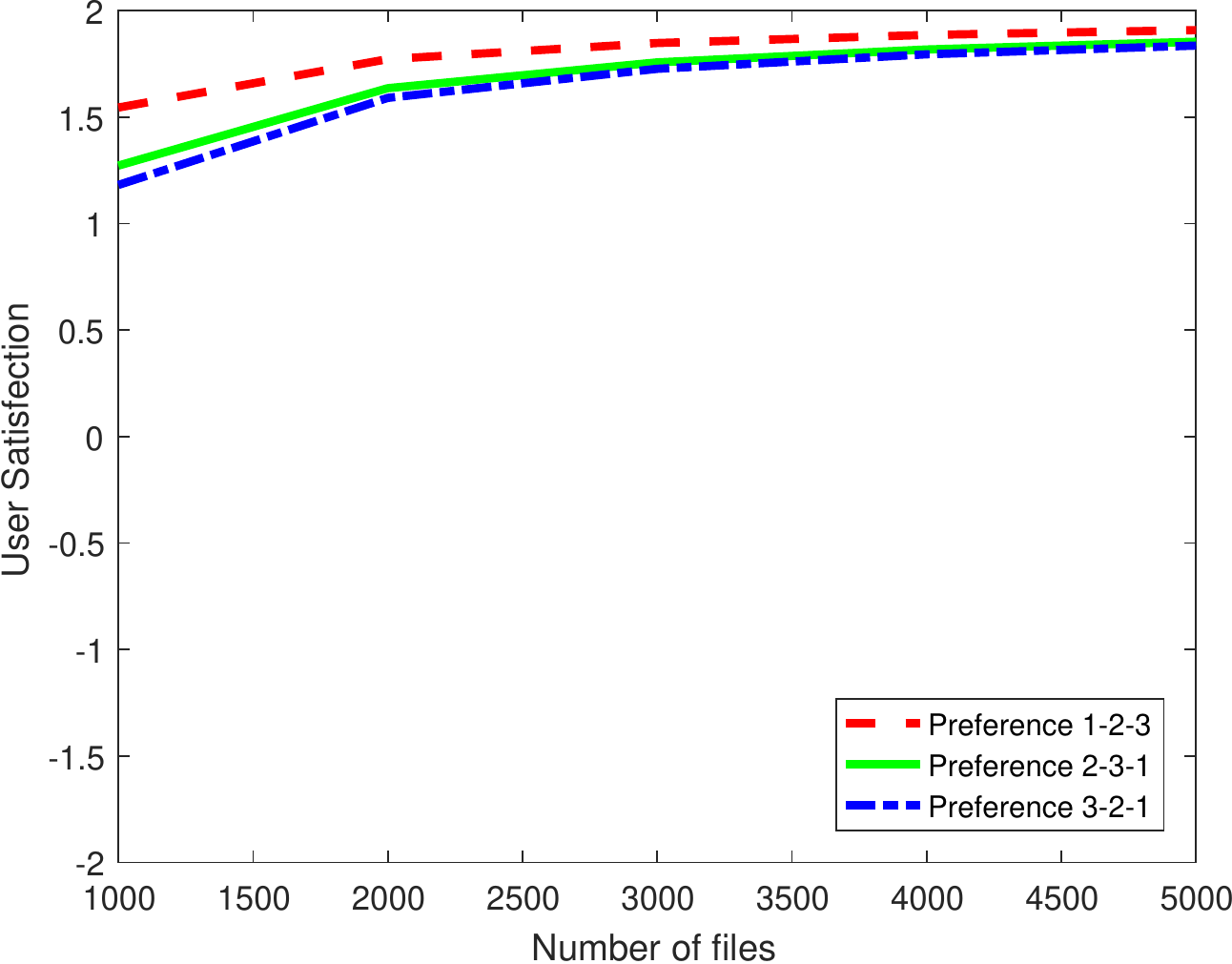}
     }
     \caption{Effect of the same parameters in different preference on user satisfaction.}
     \label{fig:UserParameters}
   \end{figure}
Figure~\ref{fig:diffPar} depicts the effects of different number of user satisfaction parameters (3, 6, 9). It can be observed from Figure~\ref{fig:diffPar} that  QoE degrades as the number of user satisfaction parameters increases. However, it can be improved using DHT-D2D based meta-data as shown in Figure~\ref{fig:compD2DwithDHT}. Figure~\ref{fig:compD2DwithDHT} illustrates  the  effects of presence of new file search (25\% and 50\%)  at each iteration. It shows that the applications having less QoE would perform considerably well under DHT-D2D after storing earlier searched files information (such as 25\% and 50\%), which distinguishes it from the other existing QoE enhancement mechanisms.   
\begin{figure}[!t]
    \centering
    \includegraphics[width=2.5in]{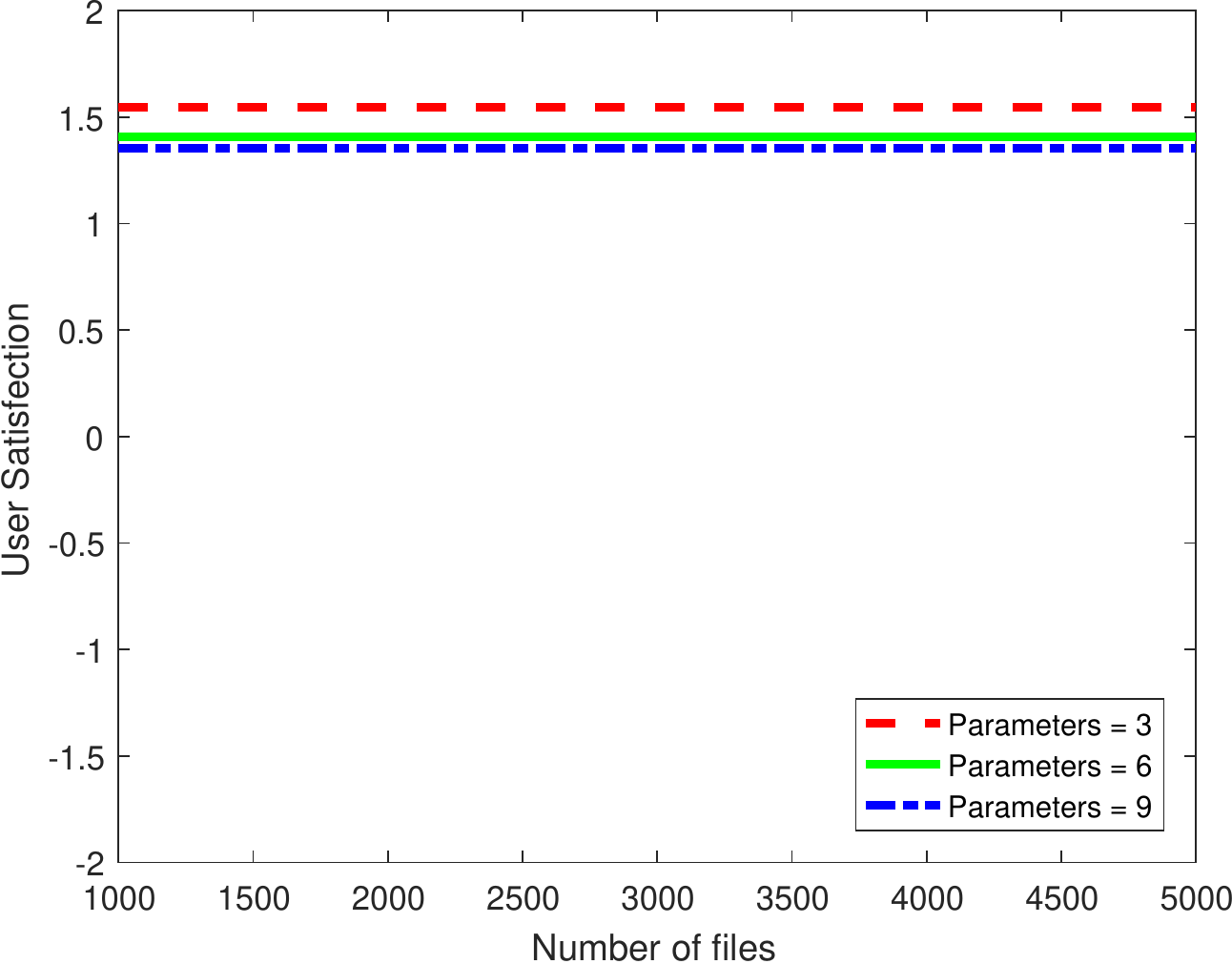}
    \caption{Effect of user satisfaction parameters (3, 6 and 9) on DHT-D2D in presence of all new files.}
    \label{fig:diffPar}
\end{figure}
\begin{figure}[!t]
    \centering
    \includegraphics[width=2.5in]{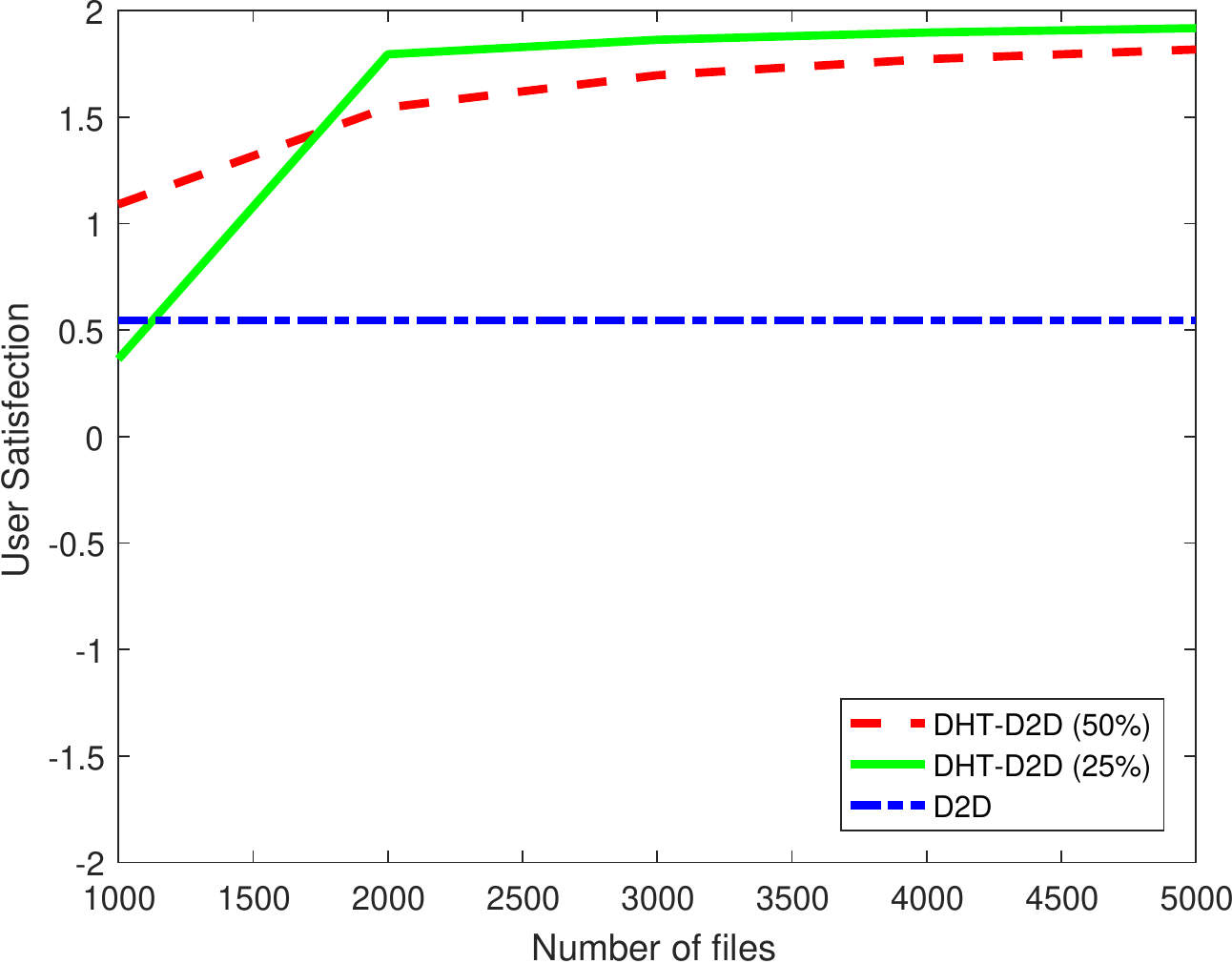}
    \caption{Effect of DHT-D2D mechanism over non-DHT D2D where old files are available for DHT within the vicinity and new files  (25\% and 50\%) are available at each iteration.}
    \label{fig:compD2DwithDHT}
\end{figure}
\section{Conclusion}
\label{con}
In this paper we evolve a mechanism to assess the quality of experience for applications which can take advantage of D2D communication by leveraging a structured overlay of user groups. First, we proposed a heuristic-based  (P-Median using Lagrangian relaxation) solutions for selection and organization of the user devices into an optimized DHT overlay. The solution has been analyzed experimentally and found to be suitable for smartphone applications. Further, we defined a model to assess, predict and reduces the QoE computation. It allows a smartphone user to define satisfaction parameters application specific and compute the QoE efficiently.   
The current work is not just limited to proof of concepts, but for creating a road-map for implementation of a whole technique for content sharing through close groups of smartphone user. This is a work in progress.

\bibliographystyle{unsrt}  
\bibliography{reference} 

\end{document}